\documentclass[12pt]{article}

\usepackage{latexsym,amsmath,amscd,amssymb,graphics,float}
\usepackage{enumerate}

\usepackage{graphicx,lscape}

\usepackage[colorlinks]{hyperref}
\usepackage{url}

\begin{document}

\newenvironment{proof}[1][Proof]{\textbf{#1.} }{\ \rule{0.5em}{0.5em}}

\newtheorem{theorem}{Theorem}[section]
\newtheorem{definition}[theorem]{Definition}
\newtheorem{lemma}[theorem]{Lemma}
\newtheorem{remark}[theorem]{Remark}
\newtheorem{proposition}[theorem]{Proposition}
\newtheorem{corollary}[theorem]{Corollary}
\newtheorem{example}[theorem]{Example}

\numberwithin{equation}{section}
\newcommand{\ep}{\varepsilon}
\newcommand{\R}{{\mathbb  R}}
\newcommand\C{{\mathbb  C}}
\newcommand\Q{{\mathbb Q}}
\newcommand\Z{{\mathbb Z}}
\newcommand{\N}{{\mathbb N}}

\newcommand{\bfi}{\bfseries\itshape}

\newsavebox{\savepar}
\newenvironment{boxit}{\begin{lrbox}{\savepar}
\begin{minipage}[b]{15.5cm}}{\end{minipage}\end{lrbox}
\fbox{\usebox{\savepar}}}

\title{{\bf A method to generate first integrals from infinitesimal symmetries}}
\author{R\u{a}zvan M. Tudoran}

\date{}
\maketitle \makeatother

\begin{abstract}
We propose a method to construct first integrals of a dynamical system, starting with a given set of linearly independent infinitesimal symmetries. In the case of two infinitesimal symmetries, a rank two Poisson structure on the ambient space it is found, such that the vector field that generates the dynamical system, becomes a Poisson vector field. Moreover, the symplectic leaves and the Casimir functions of the associated Poisson manifold are characterized. Explicit conditions that guarantee Hamilton-Poisson realizations of the dynamical system are also given. 
\end{abstract}

\medskip

\textbf{AMS 2010}:  35L65; 70H33; 70S10; 37K05.


\section{Introduction}
\label{section:one}

The explicit knowledge of first integrals of a given dynamical system has proved to be of great importance in the study of the qualitative properties of the system, e.g. integrability (see e.g., \cite{arnoldcarte}, \cite{kozlov}, \cite{kozlov1}), stability (see e.g. \cite{abraham}, \cite{verhulst}, \cite{arnoldcarte}, \cite{lasalle}, \cite{moser}). Since there exist no general methods to determine the first integrals of a given dynamical system, specific approaches need to be found for each system, or class of systems, apart. The aim of this article is to provide a method to find first integrals of a given dynamical system, if one knows a given set of infinitesimal symmetries of the system. Recall that by infinitesimal symmetry of a given dynamical system $\dot x=X(x)$ (defined on some open subset $U$ of a smooth manifold $M$), we mean a smooth vector field $Y\in\mathfrak{X}(U)$ that commutes with $X$, i.e., $[X,Y]=0$, or equivalently $\mathcal{L}_X Y=0$, where $\mathcal{L}_X$ stands for the Lie derivative along the vector field $X$.

One of the oldest results in the literature that relates the existence of a family of vector fields associated to a given dynamical system, with the existence of first integrals of the system, is due to Lie and it states that if there exists $n$ linearly independent vector fields $X_1,\dots,X_n$ on $\R^n$ that generates a solvable Lie algebra under commutation, i.e., $[X_1,X_j]=c^{1}_{1,j}X_1$, $[X_2,X_j]=c^{1}_{1,j}X_1+c^{2}_{2,j}X_2$, $\dots$, $[X_n,X_j]=c^{1}_{1,j}X_1+c^{2}_{2,j}X_2+\dots+c^{n}_{n,j}X_n$, for $j\in\{1,\dots,n\}$, where $c^{k}_{i,j}$ are the structural constants of the Lie algebra, then the differential equation $\dot{x}=X_1(x)$ is integrable by quadratures (see e.g. \cite{arnold}). Moreover, by proving these quadratures one obtain $n-1$ functionally independent first integrals of the system $\dot{x}=X_1(x)$, and consequently one get his complete integrability (see e.g. \cite{kozlov1}). This result was recently improved by Kozlov, by proving that in fact each of the differential equation $\dot{x}=X_j(x)$ is integrable by quadratures (see \cite{kozlov}).

In contrast with Lie's method, in our approach there is no relation between the dimension of the configuration manifold and the number of infinitesimal symmetries, and moreover one obtains explicitly the first integrals directly form the commutation relations.  More precisely, the method proposed in the present article, states the following. 

\bigskip
Let $\dot x= X(x)$ be a dynamical system, where $X\in\mathfrak{X}(U)$ is a smooth vector field defined eventually on an open subset $U\subseteq M$ of a smooth manifold $M$, such that the following conditions hold true:
\begin{itemize}
\item there exist $p\in\mathbb{N}$, $p>0$, and some linearly independent infinitesimal symmetries of $X$, denoted by $X_1,\dots,X_p \in \mathfrak{X}(U)$,
\item there exist smooth functions $F_{ij}^{k}\in\mathcal{C}^{\infty}(U,\mathbb{R})$, $i,j,k\in\{1,\dots,p\}$, such that for each $i,j\in\{1,\dots,p\}$:
\begin{equation*}
[X_i,X_j]=F_{ij}^{1}X_1 +\dots+ F_{ij}^{p}X_p.
\end{equation*}
\end{itemize}
Then, 
\begin{itemize}
\item for each  $i,j,k\in\{1,\dots,p\}$, the smooth function $F_{ij}^{k}\in\mathcal{C}^{\infty}(U,\mathbb{R})$, is a first integral of the dynamical system $\dot x= X(x),$
\item moreover, for each $i,j,k,l\in\{1,\dots,p\}$, the smooth function $\mathcal{L}_{X_l}F_{ij}^{k}\in\mathcal{C}^{\infty}(U,\mathbb{R})$, is a first integral of the dynamical system $\dot x= X(x)$.
\end{itemize}

\bigskip
In the case of two linearly independent infinitesimal symmetries $X_1,X_2$, the picture behind the main result can be described in terms of Poisson geometry. More precisely, in the hypothesis of the main result, the infinitesimal symmetries $X_1,X_2$, generate a bivector field $X_1 \wedge X_2 \in \mathfrak{X}^{2}(U)$ which defines a rank two Poisson structure on $U$, whose set of Casimir invariants is $\operatorname{Cas}(X_1 \wedge X_2)=\ker{\mathcal{L}_{X_1}} \cap \ker{\mathcal{L}_{X_2}}$, or equivalently, in terms of Lichnerowicz-Poisson cohomology, $H^{0}_{X_1 \wedge X_2}(U)=\ker{\mathcal{L}_{X_1}} \cap \ker{\mathcal{L}_{X_2}}$. Consequently, the vector field $X$, who generates the dynamical system $\dot x= X(x)$, becomes a Poisson vector field with respect to the Poisson structure $X_1 \wedge X_2$. Moreover, if there exists a smooth function $H\in\mathcal{C}^{\infty}(U,\mathbb{R})$ such that $X=(\mathcal{L}_{X_2}H)X_1 -(\mathcal{L}_{X_1}H)X_2$, then the dynamical system $\dot x= X(x)$ will be a Hamiltonian dynamical system modeled on the Poisson manifold $(U,X_1 \wedge X_2)$, with respect to the Hamiltonian function $H\in\mathcal{C}^{\infty}(U,\mathbb{R})$. Since the vector field $X$ is a Poisson vector field, $X$ will be also a Hamiltonian vector field, if the first Lichnerowicz-Poisson cohomology space of the Poisson manifold $(U,X_1 \wedge X_2)$ is trivial, i.e., $H^{1}_{X_1 \wedge X_2}(U)=0$. In the case when $U$ is two-dimensional, then we obtain that $H^{1}_{X_1 \wedge X_2}(U)\cong H^{1}_{dR}(U)$, and hence if $H^{1}_{dR}(U)=0$, then the dynamical system $\dot x= X(x)$ will be a Hamiltonian dynamical system on $U$.

\section{First integrals generated by infinitesimal symmetries}

In this section we provide an explicit method of constructing first integrals of a dynamical system, if one have a given set of infinitesimal symmetries. 

Let us start by recalling that by a \textit{first integral} of a dynamical system $\dot x= X(x)$, $X\in\mathfrak{X}(U)$ (defined eventually on an open subset $U\subseteq M$ of a smooth manifold $M$), we mean a smooth function $F\in\mathcal{C}^{\infty}(U,\mathbb{R})$ such that $\mathcal{L}_{X} F =0$, where the notation $\mathcal{L}_{X}$ stands for the Lie derivative along the vector field $X$. 

On the other hand, by \textit{infinitesimal symmetry} of the dynamical system $\dot x= X(x)$, we mean a smooth vector field $Y\in\mathfrak{X}(U)$ such that $[X,Y]=0$. Since $[X,Y]=\mathcal{L}_{X} Y$, one can rewrite the condition $[X,Y]=0$ in terms of Lie derivative as, $\mathcal{L}_{X} Y=0$.
 
In practice, infinitesimal symmetries are more easily found than first integrals. Nevertheless, first integrals are extremely useful for the qualitative study of the dynamical system, e.g., for the study of stability properties (see e.g., \cite{abraham}, \cite{verhulst}, \cite{arnoldcarte}, \cite{lasalle}, \cite{moser}), for the study of the integrability ( see e.g., \cite{arnoldcarte}, \cite{kozlov}, \cite{kozlov1}). Although, there are no general methods to determine the first integrals of a dynamical system.  

Let us now state the main result of this section which provides a method to construct first integrals for the dynamical system $\dot x= X(x)$ out of a given set of infinitesimal symmetries. 

\begin{theorem}\label{MT}
Let $\dot x= X(x)$ be a dynamical system, where $X\in\mathfrak{X}(U)$ is a smooth vector field defined eventually on an open subset $U\subseteq M$ of a smooth manifold $M$, such that the following conditions hold true:
\begin{enumerate}
\item there exists $p\in\mathbb{N}$, $p>0$, and some linearly independent infinitesimal symmetries of $X$, denoted by $X_1,\dots,X_p \in \mathfrak{X}(U)$,
\item there exist smooth functions $F_{ij}^{k}\in\mathcal{C}^{\infty}(U,\mathbb{R})$, $i,j,k\in\{1,\dots,p\}$, such that for each $i,j\in\{1,\dots,p\}$:
\begin{equation}\label{esy}
[X_i,X_j]=F_{ij}^{1}X_1 +\dots+ F_{ij}^{p}X_p.
\end{equation}
\end{enumerate}
Then, 
\begin{enumerate}
\item for each  $i,j,k\in\{1,\dots,p\}$, the smooth function $F_{ij}^{k}\in\mathcal{C}^{\infty}(U,\mathbb{R})$, is a first integral of the dynamical system $\dot x= X(x),$
\item moreover, for each $i,j,k,l\in\{1,\dots,p\}$, the smooth function $\mathcal{L}_{X_l}F_{ij}^{k}\in\mathcal{C}^{\infty}(U,\mathbb{R})$, is a first integral of the dynamical system $\dot x= X(x)$.
\end{enumerate}
\end{theorem} 
\begin{proof}
Let us fix $i,j\in\{1,\dots,p\}$. We will show that for each $k\in\{1,\dots,p\}$, the smooth function $F_{ij}^{k}\in\mathcal{C}^{\infty}(U,\mathbb{R})$ is a first integral of the dynamical system $\dot x= X(x)$, namely, $\mathcal{L}_{X}F_{ij}^{k}=0$.

Using the Jacobi identity of the Jacobi-Lie bracket of vector fields for the triple $X, X_i, X_j \in\mathfrak{X}(U)$ we obtain:
\begin{equation}\label{Jid}
[[X,X_i],X_j]+[[X_i,X_j],X]+[[X_j,X],X_i]=0.
\end{equation}

Since $X_i,X_j \in \mathfrak{X}(U)$ are infinitesimal symmetries of the dynamical system $\dot x= X(x)$, we obtain that the vector field $X$ commutes with both vector fields $X_i$ and $X_j$, namely $[X,X_i]=[X,X_j]=0$. Hence, the identity \eqref{Jid} becomes
\begin{equation}\label{dg}
[[X_i,X_j],X]=0,
\end{equation}
namely, $[X_i,X_j]$ is an infinitesimal symmetry of $\dot x= X(x)$, too.

Using the second relation from hypothesis, namely $[X_i,X_j]=F_{ij}^{1}X_1 +\dots+ F_{ij}^{p}X_p$, the equality \eqref{dg} can be written as follows
\begin{equation}\label{dgj}
[F_{ij}^{1}X_1,X] +\dots+ [F_{ij}^{p}X_p,X]=0.
\end{equation}

Note that for each $k\in\{1,\dots,p\}$, 
$$
[F_{ij}^{k}X_k,X]=F_{ij}^{k} [X_k,X]-(\mathcal{L}_{X}F_{ij}^{k}) X_k =-(\mathcal{L}_{X}F_{ij}^{k}) X_k.
$$

Consequently, the equality \eqref{dgj} becomes
\begin{equation*}
-(\mathcal{L}_{X}F_{ij}^{1}) X_1-\dots -(\mathcal{L}_{X}F_{ij}^{p}) X_p =0,
\end{equation*}
and since $X_1,\dots,X_p$ are linearly independent we obtain $\mathcal{L}_{X}F_{ij}^{1} =\dots =\mathcal{L}_{X}F_{ij}^{p}=0$, and hence we obtain the first part of the conclusion.

In order to prove the second part of the conclusion, let us fix $i,j,k,l\in\{1,\dots,p\}$. Since $X_l \in  \mathfrak{X}(U)$ is an infinitesimal symmetry of the dynamical system $\dot x= X(x)$, by definition we have that $[X,X_l]=0$, and consequently $\mathcal{L}_{[X,X_l]}F_{ij}^{k}=0$. 

Rewriting the last equality as
$$
\mathcal{L}_{X}(\mathcal{L}_{X_l} F_{ij}^{k}) - \mathcal{L}_{X_l}(\mathcal{L}_{X} F_{ij}^{k})=0,
$$
the conclusion follows from the fact that $\mathcal{L}_{X} F_{ij}^{k}=0$. 
\end{proof}

\medskip
Note that using the same idea as for the proof of the second statement of Theorem \eqref{MT}, for any given first integral of the dynamical system $\dot x= X(x),$ one can construct by induction a sequence of first integrals. Unfortunately, as can be seen in examples, most of these first integrals are not independent.   
\begin{remark}
Let $\dot x= X(x)$ be a dynamical system, where $X\in\mathfrak{X}(U)$ is a smooth vector field defined eventually on an open subset $U\subseteq M$ of a smooth manifold $M$. Suppose there exist $X_1,\dots,X_p \in \mathfrak{X}(U)$, infinitesimal symmetries of $X$. 

If $F\in\mathcal{C}^{\infty}(U,\mathbb{R})$ is an arbitrary given first integral of the dynamical system $\dot x= X(x),$ then to each sequence $(a_n)_{n\in\mathbb{N}}\in\{1,\dots,p\}^{\mathbb{N}}$, one can associate a sequence of smooth functions $(F_n)_{n\in\mathbb{N}}\subset\mathcal{C}^{\infty}(U,\mathbb{R})$, given by $$F_n=\mathcal{L}_{X_{a_n}}(\dots(\mathcal{L}_{X_{a_1}}( \mathcal{L}_{X_{a_0}}F))\dots),$$ whose terms are first integrals of the dynamical system $\dot x= X(x).$
\end{remark}

\bigskip
Let us now point out some properties regarding the geometric picture associated with the Theorem \eqref{MT}. Before stating the results, let us recall some notations. Note that if $X_1,\dots,X_p \in  \mathfrak{X}(U)$ are linearly independent smooth vector fields defined on the open subset $U\subseteq M$ of a smooth manifold $M$, then the smooth assignment $$x\in U \mapsto \langle X_1 (x),\dots,X_p (x)\rangle:=\operatorname{span}_{\mathbb{R}}\{X_1(x),\dots,X_p (x)\}\subseteq T_x U,$$
determines a $p-$dimensional smooth distribution on $U$, denoted by $\langle X_1 ,\dots,X_p \rangle$.
\begin{remark}\label{RMT}
In the hypothesis of Theorem \eqref{MT}, the Frobenius theorem implies the integrability of the $p-$dimensional distribution $\langle X_1 ,\dots,X_p \rangle$, generated by the infinitesimal symmetries $X_1 ,\dots,X_p \in \mathfrak{X}(U)$ of the dynamical system $\dot x =X(x)$, $x\in U$.

Let us now distinguish between two cases according if the smooth vector fields $X$, $X_1$, $\dots$, $X_p$ are, or are not, linearly independent on the open subset $U\subseteq M$. 
\begin{itemize}
\item If $X$ is linearly dependent with respect to $X_1,\dots,X_p$, then
 $\langle X, X_1 ,\dots,X_p \rangle = \langle X_1 ,\dots,X_p \rangle$, and hence the vector field $X$ is tangent to the leaves of the foliation induced by the integrable distribution generated by the infinitesimal symmetries $X_1 ,\dots,X_p$.
\item If $X, X_1 ,\dots,X_p$, are linearly independent, then by using again the Frobenius theorem, we obtain that $\langle X, X_1 ,\dots,X_p \rangle$ is a $(p+1)-$dimensional integrable distribution on $U$. In this case, for each $x\in U$ there exists an open subset $U_x \subseteq U$ such that the following diffeomorphic splitting holds
$$
\Sigma_{x}^{\langle X, X_1 ,\dots,X_p \rangle}\cap U_{x} \cong \left(\Sigma_{x}^{\langle X_1 ,\dots,X_p \rangle}\cap U_{x}\right) \times \left(\Sigma_{x}^{\langle X \rangle}\cap U_{x}\right),
$$
where $\Sigma_{x}^{D}$ stands for the connected component which contains $x$, of the leaf passing through $x$, of the foliation induced by the integrable distribution $D$.
\end{itemize}
\end{remark}

\section{Some generalities on Poisson manifolds}

Since the rest of this work is concerned with the description of the geometric picture behind the Theorem \eqref{MT} for the special case of two linearly independent infinitesimal symmetries, and the associated geometry proves to be the Poisson geometry, in order to have a self contained presentation, we dedicate this section to recall some generalities on Poisson geometry. The results presented here follows \cite{vaisman}, \cite{dufur}, \cite{naka}, \cite{MRT}, \cite{ratiurazvan}.

Let $M$ be an $n-$dimensional smooth manifold. Denote by $$L(M)=\left(\oplus_{i=0}^{n}\mathfrak{X}^{i}(M),\wedge\right)$$ the contravariant Grassmannn algebra of $M$, where for each $i\in\{0,\dots, n\}$, $\mathfrak{X}^{i}(M)$ stands for the space of smooth $i-$vector fields of $M$. Recall that by convention $\mathfrak{X}^{0}(M)=\mathcal{C}^{\infty}(M,\mathbb{R})$. 

Recall that the Schouten-Nijenhuis bracket is a generalization of the Jacobi-Lie bracket of ($1-$)vector fields, and it is a homogeneous bi-derivation of degree $-1$ defined on $L(M)$
$$(T,S)\in \mathfrak{X}^{i}(M) \times \mathfrak{X}^{j}(M) \mapsto [T,S]\in \mathfrak{X}^{i+j-1}(M),$$
characterized by the following properties:
\begin{enumerate}
\item $[F,G]=0$ for every $F,G \in \mathfrak{X}^{0}(M)$,
\item $[X,F]=\mathcal{L}_{X}F$ for every $F\in \mathfrak{X}^{0}(M)$, $X\in \mathfrak{X}^{1}(M)$,
\item $[X,Y]=\mathcal{L}_{X}Y$ for every $X,Y\in \mathfrak{X}^{1}(M)$,
\item $[T,S\wedge R]=[T,S]\wedge R + (-1)^{(t-1)s}S\wedge [T,R]$, for every $T\in \mathfrak{X}^{t}(M)$, $S\in \mathfrak{X}^{s}(M)$, $R\in \mathfrak{X}^{r}(M)$,
\item $[T,S]=(-1)^{(t-1)(s-1)+1}[S,T]$, for every $T\in \mathfrak{X}^{t}(M)$, $S\in \mathfrak{X}^{s}(M)$,
\item $(-1)^{(t-1)(r-1)}[T,[S,R]]+(-1)^{(s-1)(t-1)}[S,[R,T]]+(-1)^{(r-1)(s-1)}[R,[T,S]]=0$, for every $T\in \mathfrak{X}^{t}(M)$, $S\in \mathfrak{X}^{s}(M)$, $R\in \mathfrak{X}^{r}(M)$.
\end{enumerate}

The Schouten-Nijenhuis bracket provides a simple characterization of Poisson structures. Let us recall that a \textit{Poisson bracket} on a smooth manifold $M$, is an $\mathbb{R}-$bilinear map $$(F,G)\in\mathcal{C}^{\infty}(M,\mathbb{R})\times \mathcal{C}^{\infty}(M,\mathbb{R}) \mapsto \{F,G\} \in \mathcal{C}^{\infty}(M,\mathbb{R}),$$ which is also skew-symmetric, verifies the Jacobi identity, and respectively the Leibniz rule of derivation. In other words, the algebra $\mathcal{C}^{\infty}(M,\mathbb{R})$ becomes a \textit{Poisson algebra}, i.e., a commutative algebra together with a Lie bracket that also satisfies the Leibniz rule of derivation.

Note that for each arbitrary smooth function $F\in\mathcal{C}^{\infty}(M,\mathbb{R})$, since the Poisson bracket verifies the Leibniz rule of derivation, there exists a smooth vector field, $X_F$, called the \textit{Hamiltonian vector filed of $F$}, such that
$$\{F,G\}=\mathcal{L}_{X_F} G = - \mathcal{L}_{X_G} F =\langle dG, X_F \rangle =-\langle dF, X_G \rangle,$$ for every smooth function $G\in\mathcal{C}^{\infty}(M,\mathbb{R})$, where $\langle\cdot,\cdot\rangle$ stands for the natural pairing between smooth one-forms and smooth vector fields, i.e., $\langle \alpha, X \rangle=\alpha\cdot X = \alpha(X)$, for any smooth one-form $\alpha$ and respectively any smooth vector field $X$.

Hence, the above relation implies that the Poisson bracket is determined by a smooth bivector field, $\Pi\in\mathfrak{X}^{2}(M)$, via the natural pairing between smooth two-forms and bivector fields, namely for every $F,G\in \mathcal{C}^{\infty}(M,\mathbb{R})$
\begin{equation}\label{PB}
\{F,G\}=\langle dF \wedge dG, \Pi \rangle.
\end{equation}
The bivector field $\Pi$ it is called \textit{Poisson bivector field}, or \textit{Poisson structure}. Note that due to the characteristic properties of the Poisson bracket, a Poisson bivector field need to meet some special requirements. The pair $(M,\Pi)$, consisting of a manifold $M$, together with a Poisson structure $\Pi$, it is called \textit{Poisson manifold}. Given a Poisson structure $\Pi$ on $M$, one can associate a natural homomorphism $$\alpha\in \Gamma(T^{\star}M) \mapsto \Pi^{\sharp}(\alpha)\in \Gamma(TM),$$ where $\langle \alpha \wedge \beta, \Pi \rangle=\langle \beta, \Pi^{\sharp}(\alpha)\rangle$, for any one-form $\beta\in \Gamma(T^{\star}M)$. Recall that the  \textit{rank} of the Poisson structure $\Pi$ at a point $x\in M$, $\operatorname{rank}(\Pi(x))$, is the dimension of the vector space $\operatorname{Im}(\Pi^{\sharp}_{x})\subseteq T_x M$. Note that for any $F\in \mathcal{C}^{\infty}(M,\mathbb{R})$, the associated Hamiltonian vector field, can be written as $X_F = \Pi^{\sharp}(dF)$.

A special class of Poisson manifolds are the \textit{symplectic manifolds}. Recall that a symplectic manifold is a pair $(M,\omega)$, consisting of a smooth manifold $M$ equipped with a nondegenerate closed differential two-form $\omega$, called the \textit{symplectic form}. Since the nondegeneracy of $\omega$ means that the homomorphism $$ X\in \Gamma(TM) \mapsto \omega^{\flat}(X):=i_{X}\omega \in \Gamma(T^{\star}M)$$ is an isomorphism, we get that for each smooth function $F\in \mathcal{C}^{\infty}(M,\mathbb{R})$, one can associate an unique vector field, denoted by $X_F$, and called its \textit{Hamiltonian vector field}, as follows: $i_{X_F}\omega = dF$. 

If one define for every $F,G\in \mathcal{C}^{\infty}(M,\mathbb{R})$, $\{F,G\}_{\omega}:=\omega (X_F,X_G),$
then $\{\cdot,\cdot\}_{\omega}$ is a Poisson bracket on the symplectic manifold $(M,\omega)$, and hence $(M,\{\cdot,\cdot\}_{\omega})$ becomes a Poisson manifold.

Let us now come back to the general case of Poisson manifolds. In terms of Schouten-Nijenhuis bracket, an arbitrary bivector field $\Pi\in\mathfrak{X}^{2}(M)$, it is a Poisson bivector field, if and only if $[\Pi,\Pi]=0$. Moreover, if $(M,\Pi)$ is a Poisson manifold, then the Hamiltonian vector field generated by a smooth function $F\in \mathcal{C}^{\infty}(M,\mathbb{R})$, can be expressed in terms of Schouten-Nijenhuis bracket as, $X_F =[\Pi,F]$.

Denote by $D_{\Pi}:L(M)\rightarrow L(M)$, the linear operator given by $$S\in\mathfrak{X}^{i}(M)\mapsto [\Pi,S]\in\mathfrak{X}^{i+1}(M)$$ for each $i\in\{0,\dots,n\}$. If $\Pi$ is a Poisson structure, i.e., $[\Pi,\Pi]=0$, the operator $D_{\Pi}$ satisfies $D_{\Pi}^{2} =0$, and hence becomes a coboundary operator. The associated cohomology it is called \textit{Lichnerowicz-Poisson cohomology}, and is denoted by $H^{\star}_{\Pi}(M)$. Recall that the $k-$th Lichnerowicz-Poisson cohomology space of the Poisson manifold $(M,\Pi)$ is given by
$$
H^{k}_{\Pi}(M)={\ker(D_{\Pi}:\mathfrak{X}^{k}(M)\rightarrow \mathfrak{X}^{k+1}(M))}/{\operatorname{Im}(D_{\Pi}:\mathfrak{X}^{k-1}(M)\rightarrow \mathfrak{X}^{k}(M))}.
$$

Recall also that for a given Poisson manifold, $(M,\Pi)$, if $\Pi$ is of full rank (and hence the isomorphism $-(\Pi^{\sharp})^{-1}$ generates a symplectic structure), then $H^{\star}_{\Pi}(M)\cong H^{\star}_{dR}(M)$, where $H^{\star}_{dR}(M)$ stands for the classical de Rham cohomology.

Given a Poisson structure, $\Pi\in \mathfrak{X}^{2}(M)$, the space of its infinitesimal automorphisms, namely the space of vector fields $X$, such that $\mathcal{L}_{X} \Pi =0$, (or equivalently $D_{\Pi}(X)=0$), it is denoted by $Z^{1}_{\Pi}(M)\subseteq\mathfrak{X}^{1}(M)$, and its elements are called \textit{Poisson vector fields} of the Poisson manifold $(M,\Pi)$.

If one denotes by $B^{1}_{\Pi}(M)\subseteq\mathfrak{X}^{1}(M)$ the space of \textit{Hamiltonian vector fields} of the Poisson manifold $(M,\Pi)$, i.e., the space of vector fields $X\in\mathfrak{X}^{1}(M)$ for which there exists $F\in \mathfrak{X}^{0}(M)$ such that $X=D_{\Pi}(F)$, then the first Lichnerowicz-Poisson cohomology space can be written as
$$
H^{1}_{\Pi}(M)\cong Z^{1}_{\Pi}(M)/B^{1}_{\Pi}(M).
$$
Similarly, $H^{0}_{\Pi}(M)$ is the center of the Poisson algebra $\mathcal{C}^{\infty}(M,\mathbb{R})=\mathfrak{X}^{0}(M)$, and is called the \textit{space of Casimir functions} of the Poisson manifold $(M,\Pi)$.

Recall that any Poisson manifold $(M,\Pi)$ admits an integrable generalized distribution given by the assignment 
$$
x\in M \mapsto \operatorname{Im}(\Pi^{\sharp}_{x})=\{v_x \in T_x M \mid (\exists) X\in B^{1}_{\Pi}(M), s.t. \ X(x)=v_x\}\subseteq T_x M,
$$
whose leaves are simplectic manifolds (of possible different dimensions), and are called the \textit{symplectic leaves} of the Poisson manifold $(M,\Pi)$. Recall that the dimension of a symplectic leaf $\Sigma_x \subset M$ is equal to $\dim_{\mathbb{R}}\operatorname{Im}(\Pi^{\sharp}_{x})=\operatorname{rank}(\Pi(x))$.

\section{Two infinitesimal symmetries and the associated Poisson geometry}

In this section we study in detail the geometric picture behind the Theorem \eqref{MT}, for the case of two linearly independent infinitesimal symmetries, $X_1,X_2\in\mathfrak{X}(U)$, of the dynamical system $\dot x= X(x)$, $X\in\mathfrak{X}(U)$, where $U$ is an open subset of a smooth manifold $M$. 

Let us now state the first result of this section, which shows that any pair of linearly independent smooth vector fields, $X_1,X_2\in\mathfrak{X}(U)$, for which there exist $F_1,F_2 \in \mathcal{C}^{\infty}(U,\mathbb{R})$ such that $[X_1,X_2]=F_1 X_1 +F_2 X_2$, generates a Poisson structure on $U$.

\begin{proposition}\label{Pstr}
Let $M$ be a smooth manifold and $U\subseteq M$ an open subset. If $X_1,X_2\in\mathfrak{X}(U)$ is a pair of linearly independent smooth vector fields for which there exist $F_1,F_2 \in \mathcal{C}^{\infty}(U,\mathbb{R})$ such that $[X_1,X_2]=F_1 X_1 +F_2 X_2$, then the bivector field $\Pi_{X_1,X_2}:=X_1\wedge X_2\in\mathfrak{X}^{2}(U)$ is a Poisson structure on $U$.
\end{proposition}
\begin{proof}
Using the characterization of Poisson structures in terms of Schouten-Nijenhuis bracket, the bivector field $\Pi_{X_1,X_2}\in\mathfrak{X}^{2}(U)$ is a Poisson structure if and only if $$[\Pi_{X_1,X_2},\Pi_{X_1,X_2}]=0.$$

A simple computation shows that $$[\Pi_{X_1,X_2},\Pi_{X_1,X_2}]=2X_1\wedge X_2 \wedge [X_1,X_2],$$ and hence the conclusion follows from the existence of the smooth functions $F_1,F_2 \in \mathcal{C}^{\infty}(U,\mathbb{R})$ such that $[X_1,X_2]=F_1 X_1 +F_2 X_2$.
\end{proof}

Let us now analyze the geometry of $U$, as a Poisson manifold together with the Poisson structure $\Pi_{X_1,X_2}=X_1\wedge X_2$. Next result gives an explicit formula for the Poisson bracket introduced by $\Pi_{X_1,X_2}$.

\begin{proposition}\label{PBRT}
The Poisson bracket introduced by the Poisson structure $\Pi_{X_1,X_2}$ is given by
\begin{equation}\label{EPB}
\{F,G\}_{\Pi_{X_1,X_2}}=(\mathcal{L}_{X_1}F)(\mathcal{L}_{X_2}G) - (\mathcal{L}_{X_2}F)(\mathcal{L}_{X_1}G),
\end{equation}
for any $F,G\in \mathcal{C}^{\infty}(U,\mathbb{R})$.
\end{proposition}
\begin{proof}
Using the formula \eqref{PB} in the case of the Poisson structure $\Pi_{X_1,X_2}=X_1\wedge X_2$, we obtain directly 
\begin{align*}
\{F,G\}_{\Pi_{X_1,X_2}}&=\langle dF \wedge dG, \Pi_{X_1,X_2} \rangle = \langle dF \wedge dG, X_1\wedge X_2 \rangle\\
&=\left| 
\begin{array}{cc}  
\langle dF,X_1\rangle & \langle dF,X_2\rangle \\
\langle dG,X_1\rangle & \langle dG,X_2\rangle \\
\end{array} \right|
=\left| 
\begin{array}{cc}  
 dF \cdot X_1 & dF \cdot X_2 \\
dG \cdot X_1  & dG \cdot X_2 \\
\end{array} \right| \\
&=(\mathcal{L}_{X_1}F)(\mathcal{L}_{X_2}G) - (\mathcal{L}_{X_2}F)(\mathcal{L}_{X_1}G).
\end{align*}
\end{proof}

Let us now provide an explicit formula the Hamiltonian vector fields with respect to the Poisson structure $\Pi_{X_1,X_2}$.

\begin{proposition}\label{HVF}
Let $F\in \mathcal{C}^{\infty}(U,\mathbb{R})$ be a smooth function on the Poisson manifold $(U,\Pi_{X_1,X_2})$. Then the Hamiltonian vector field associated with $F$ is given by
\begin{equation*}
X_{F}=(\mathcal{L}_{X_2}F)X_1 - (\mathcal{L}_{X_1}F)X_2.
\end{equation*}
\end{proposition}
\begin{proof}
Using the characterization of Hamiltonian vector fields via Schouten-Nijenhuis bracket, we obtain
\begin{align*}
X_{F}&=[\Pi_{X_1,X_2},F]=[X_1\wedge X_2 ,F]=[F,X_1\wedge X_2]\\
&=[F,X_1]X_2 - [F,X_2]X_1 = [X_2, F]X_1 - [X_1,F]X_2\\
&=(\mathcal{L}_{X_2}F)X_1 - (\mathcal{L}_{X_1}F)X_2.
\end{align*}
\end{proof}

A consequence of the above result is the characterization of the space of Casimir functions of the Poisson manifold $(U,\Pi_{X_1,X_2})$. Before stating the result, let us recall that a smooth function $C\in \mathcal{C}^{\infty}(U,\mathbb{R})$, is a Casimir function of the Poisson manifold $(U,\Pi_{X_1,X_2})$ if and only if it belongs to the center of the Poisson algebra $\mathcal{C}^{\infty}(U,\mathbb{R})$, i.e., $\{C,F\}_{\Pi_{X_1,X_2}}=0$, for every $F\in \mathcal{C}^{\infty}(U,\mathbb{R})$, or equivalently, $X_C=0$.

\begin{proposition}\label{TCPX}
The space of Casimir functions of the Poisson manifold $(U,\Pi_{X_1,X_2})$ is $\operatorname{Cas}(\Pi_{X_1,X_2})=\ker\mathcal{L}_{X_1} \bigcap \ker\mathcal{L}_{X_2}$, or equivalently, in terms of Lichnerowicz-Poisson cohomology, $H^{0}_{\Pi_{X_1,X_2}}(U)=\ker\mathcal{L}_{X_1} \bigcap \ker\mathcal{L}_{X_2}$.
\end{proposition}
\begin{proof}
Recall that $C\in  \mathcal{C}^{\infty}(U,\mathbb{R})$ is a Casimir function of the Poisson manifold $(U,\Pi_{X_1,X_2})$ if and only if $X_C=0$. By using the Proposition \eqref{HVF}, the equation $X_C=0$ is equivalent to
\begin{equation*}
(\mathcal{L}_{X_2}C)X_1 - (\mathcal{L}_{X_1}C)X_2 =0.
\end{equation*}
Since by definition, the vector fields $X_1$ and $X_2$ are linearly independent, we obtain that $$\mathcal{L}_{X_1}C=\mathcal{L}_{X_2}C =0,$$ and hence the conclusion.
\end{proof}

Next result gives a characterization of the symplectic leaves of the Poisson manifold $(U,\Pi_{X_1,X_2})$.
\begin{proposition}\label{SYML}
The Poisson structure $\Pi_{X_1,X_2}$ has constant rank two on $U$. Moreover, if $\Sigma_{x_0}\subseteq U$ is the symplectic leaf through the point $x_{0}\in U$, then for each $x\in \Sigma_{x_0}$, $$T_x \Sigma_{x_0}=\operatorname{span}_{\mathbb{R}}\{X_1(x),X_2(x)\}.$$
\end{proposition}
\begin{proof}
The proof is adapted from \cite{MRT}, where the case of a Poisson structure generated by two commuting vector fields $X_1, X_2$ is presented. Let $x_0 \in U$ be fixed. Since $X_1$ and $X_2$ are linearly independent, they do not have equilibrium points inside $U$. Hence, one can find a local chart around $x_0$ (whose geometric zone will be denoted by $V\subseteq  U$), such that $X_1=\partial_{x_1}$, where $(x_1,\dots,x_n)$ stands for the local coordinates on $V$. Let $X_2= X_{2}^{1}\partial_{x_1}+\dots+ X_{2}^{n}\partial_{x_n}$, be the local expression of the vector field $X_2$, where $X_{2}^{1},\dots,X_{2}^{n}\in\mathcal{C}^{\infty}(V,\mathbb{R})$. 

Hence, the local expression of the Poisson structure $\Pi_{X_1,X_2}$ on $V$, it is given for each $x\in V$ by
$$\Pi_{X_1,X_2}(x)=\sum_{1\leq i<j \leq n}\{x_i,x_j\}_{\Pi_{X_1,X_2}}(x)\partial_{x_i}\wedge \partial_{x_j}.$$ 

In order to compute the rank of $\Pi_{X_1,X_2}$, recall from \eqref{EPB} that
\begin{align*}
\{x_i,x_j\}_{\Pi_{X_1,X_2}}&= (\mathcal{L}_{X_1}x_i)(\mathcal{L}_{X_2}x_j) - (\mathcal{L}_{X_2}x_i)(\mathcal{L}_{X_1}x_j)\\
&=\delta_{1i}X_{2}^{j}-X_{2}^{i}\delta_{1j},
\end{align*}
where $\delta$ stands for Kronecker's delta, namely, $\delta_{ij}=0$, if $i\neq j$, and respectively $\delta_{ij}=1$, if $i=j$.

Since $X_1$ and $X_2$ are linearly independent on $V$ too, we obtain that $$\operatorname{rank} \Pi_{X_1,X_2}(x)=2$$ for any $x\in V$. As $x_0\in U$ was arbitrary chosen, one conclude that $\operatorname{rank} \Pi_{X_1,X_2}(x_0)=2$ for every $x_0 \in U$. Hence, each symplectic leaf of the Poisson manifold $(U,\Pi_{X_1,X_2})$ is two-dimensional. 

Let us fix a symplectic leaf $\Sigma_{x_0}$. We shall prove now that for each $x\in \Sigma_{x_0}$, $$T_x \Sigma_{x_0}=\operatorname{span}_{\mathbb{R}}\{X_1(x),X_2(x)\}.$$ 

If one shows that $$T_x \Sigma_{x_0} \subseteq \operatorname{span}_{\mathbb{R}}\{X_1(x),X_2(x)\}, $$ then the conclusion follows from the fact that both vector spaces are two dimensional. In order to prove the inclusion, let $v_x \in T_x \Sigma_{x_0}$. Then, since $\Sigma_{x_0}$ is a symplectic leaf, there exists a Hamiltonian vector field $X_{H}$, such that $v_x=X_H (x)$. By Proposition \eqref{HVF}, we get that $$v_x = (\mathcal{L}_{X_2}H)(x)X_1 (x) - (\mathcal{L}_{X_1}H)(x)X_2 (x),$$ and hence $$v_x \in \operatorname{span}_{\mathbb{R}}\{X_1(x),X_2(x)\}.$$
\end{proof}

Let us now prove the main result of this section, which describes the geometry of the class of dynamical systems that admits two linearly independent infinitesimal symmetries. 
\begin{theorem}\label{PVF}
Let $\dot x= X(x)$ be a dynamical system, where $X\in\mathfrak{X}(U)$ is a smooth vector field defined eventually on an open subset $U\subseteq M$ of a smooth manifold $M$, such that the following conditions hold true:
\begin{enumerate}
\item there exist two linearly independent infinitesimal symmetries of $X$, denoted by $X_1,X_2 \in \mathfrak{X}(U)$,
\item there exist two smooth functions $F_1,F_2 \in\mathcal{C}^{\infty}(U,\mathbb{R})$, such that $[X_1,X_2]=F_1 X_1 + F_2 X_2$.
\end{enumerate}
Then, apart from the conclusions of Theorem \eqref{MT}, we obtain that
\begin{enumerate}
\item the bivector field $\Pi_{X_1,X_2}:=X_1\wedge X_2$ is a Poisson structure on the open subset $U$,
\item the vector field $X$ becomes a Poisson vector field of the Poisson manifold $(U,\Pi_{X_1,X_2})$, or equivalently the Poisson structure $\Pi_{X_1,X_2}$ is constant along the flow of $X$, (and hence the flow of $X$ maps symplectic leaves to symplectic leaves),
\item if moreover there exists a smooth function $H\in\mathcal{C}^{\infty}(U,\mathbb{R}) $ such that 
$$
X=(\mathcal{L}_{X_2}H)X_1 - (\mathcal{L}_{X_1}H)X_2,
$$
then the vector field $X$ becomes a Hamiltonian vector field. In this case, together with $H$, every function from the set $\ker\mathcal{L}_{X_1} \bigcap \ker\mathcal{L}_{X_2},$ is a first integral of the dynamical system $\dot x= X(x)$.
\end{enumerate}
\end{theorem}
\begin{proof}
The first statement follows directly from  Proposition \eqref{Pstr}, since the infinitesimal symmetries $X_1$ and $X_2$ are linearly independent.

\bigskip
The second statement is equivalent with the condition that the Lie derivative of the Poisson structure $\Pi_{X_1,X_2}$ is zero along the vector field $X$, i.e., $\mathcal{L}_{X}\Pi_{X_1,X_2}=0$, or equivalently in terms of Schouten-Nijenhuis bracket, $[\Pi_{X_1,X_2},X]=0$.

The relation $[\Pi_{X_1,X_2},X]=0$ follows by applying the properties of the Schouten-Nijenhuis bracket, namely
\begin{align*}
[X_1\wedge X_2,X]&=-[X,X_1\wedge X_2]= -[X,X_1]\wedge X_2 - X_1 \wedge[X,X_2]\\
&=0,
\end{align*}
since $[X,X_1]=[X,X_2]=0$, because $X_1,X_2$ are infinitesimal symmetries of $X$.

\bigskip
The third statement is a direct consequence of Proposition \eqref{TCPX}, and respectively Proposition \eqref{HVF}. In order to complete the proof, note that if $C$ is a Casimir function of the Poisson bracket $\Pi_{X_1,X_2}$, and $X$ is a Hamiltonian vector field associated with the smooth function $H\in \mathcal{C}^{\infty}(U,\mathbb{R})$, i.e. $X=X_H$, then $C$ is a first integral for $X$. Indeed, since by definition, Casimir functions Poisson commute with all functions from $\mathcal{C}^{\infty}(U,\mathbb{R})$, we obtain that
$$
\mathcal{L}_{X}C=\mathcal{L}_{X_H}C=\{H,C\}_{\Pi_{X_1,X_2}}=0,
$$
and hence $C$ is a first integral of $X$.

Obviously, $H$ is also a first integral of $X$ since $\mathcal{L}_{X}H=\mathcal{L}_{X_H}H=\{H,H\}_{\Pi_{X_1,X_2}}=0$.
\end{proof}

\begin{remark}
In the hypothesis of Theorem \eqref{PVF}, if $H^{1}_{\Pi_{X_1,X_2}}(U)=0$, then the vector field $X$, as it is a Poisson vector filed, it must be also a Hamiltonian vector field of the Poisson manifold $(U,\Pi_{X_1,X_2})$. Moreover, if the smooth manifold $M$ is two dimensional, then the Poisson structure $\Pi_{X_1,X_2}$ is of full rank on the open set $U$, and hence $H^{1}_{\Pi_{X_1,X_2}}(U)\cong H^{1}_{dR}(U)$.
\end{remark}

Let us now illustrate the main results in the case of particular dynamical system. Recall from \cite{tudoranjgp} that this dynamical system represents the (local) normal form of two-dimensional completely integrable systems (restricted to a certain open and dense set).
\begin{example}
Let $\dot{\mathbf{x}}=X(\mathbf{x})$, be the dynamical system generated by the vector field 
\begin{equation*}
X(\mathbf{x})=x\partial_x +y \partial_y, \  \mathbf{x}=(x,y)\in\mathbb{R}^2.
\end{equation*}

If one restrict the vector field $X$ to the open set $$\mathcal{O}=\mathbb{R}^{2} \setminus\left\{ \{(x,0)\mid x\in\mathbb{R}\} \cup \{(0,y)\mid y\in\mathbb{R}\} \right\},$$ then the dynamical system $\dot{\mathbf{x}}=X(\mathbf{x})$ admits two linearly independent infinitesimal symmetries, namely $X_1 (x,y)=y\partial_x + x\partial_y$, and respectively $X_2 (x,y)= (y^{2} / x) \partial_x$. 

Indeed, one can easily verify that $[X,X_1]=[X,X_2]=0$, and respectively, 
$$
\det 
\left[
\begin{array}{cc}
y & x \\
y^2 / x & 0
\end{array}
\right]=-y^2\neq 0,$$ for every $(x,y)\in \mathcal{O}$.

Moreover, there exist $F_1, F_2 \in \mathcal{C}^{\infty}(\mathcal{O},\mathbb{R})$, $F_1 (x,y)=y^2 / x^2$, $F_2 (x,y)=2 x/y$ such that 
$$
[X_1,X_2]=F_1 X_1 + F_2 X_2.
$$

\medskip
Hence, from Theorem \eqref{MT}, the functions $F_1,F_2, \mathcal{L}_{X_1}F_1, \mathcal{L}_{X_1}F_2, \mathcal{L}_{X_2}F_1, \mathcal{L}_{X_2}F_2 $ are first integrals of the dynamical system $\dot{\mathbf{x}}=X(\mathbf{x})$. After straightforward computations we obtain that all these first integrals are functional combinations of the first integral $H:=(1/2) F_2$.

\medskip
Since we analyze the problem from the dynamical point of view, we will restrict from now on to a fixed connected component of $\mathcal{O}$, e.g., $U=(0,\infty)\times (0,\infty)\subset \mathcal{O}$. 

\medskip
From Proposition \eqref{Pstr} we obtain that $\Pi_{X_1,X_2}= X_1\wedge X_2$ is a Poisson structure on $U$. More precisely, 
$$
\Pi_{X_1,X_2}(x,y)= - y^2 \partial_x \wedge \partial_y,
$$
for every $(x,y)\in U$. Consequently, from Proposition \eqref{PBRT}, the associated Poisson bracket, $\{\cdot,\cdot\}_{\Pi_{X_1,X_2}}$, is given by
\begin{align*}
\{F,G\}_{\Pi_{X_1,X_2}}(x,y)=-y^2(\partial_x F \partial_y G -\partial_y F \partial_x G),
\end{align*}
for any $F,G\in \mathcal{C}^{\infty}(U,\mathbb{R})$.

Applying the Theorem \eqref{PVF} to the dynamical system, $\dot{\mathbf{x}}=X(\mathbf{x})$, one obtains that $X$ is a Poisson vector filed on the Poisson manifold $(U,\Pi_{X_1,X_2})$. Indeed, one can also verify directly that the Poisson structure $\Pi_{X_1,X_2}$ is constant along the flow of $X$, i.e., $\varphi_{t}^{X}:(x,y)\in U \mapsto \varphi_{t}^{X}(x,y)=(e^{t}x,e^{t}y)\in U$, for every $t\in\mathbb{R}$, since 
\begin{align*}
\Pi_{X_1,X_2}(e^{t}x,e^{t}y)&= - \left( e^{t}y \right)^{2} \partial_{(e^{t}x)} \wedge \partial_{(e^{t}y)}=-e^{2t}y^{2}e^{-2t}\partial_{x} \wedge \partial_{y}= - y^2 \partial_x \wedge \partial_y \\
&= \Pi_{X_1,X_2}(x,y),
\end{align*}
for every $(x,y)\in U$ and respectively every $t\in \mathbb{R}$.

Since $\operatorname{rank}\Pi_{X_1,X_2}(x,y)=2$, for any $(x,y)\in U$, the Poisson structure is of full rank on $U$, and hence it defines a symplectic structure, given by $$\omega_{X_1,X_2}(x,y)= - y^{-2} dx\wedge dy,$$ for any $(x,y)\in U$. Consequently $H^{1}_{\Pi_{X_1,X_2}}(U)\cong H^{1}_{dR}(U)=0$, and hence the Poisson vector field $X$ must be also a Hamiltonian vector field on the Poisson manifold $(U,\Pi_{X_1,X_2})$. Indeed, if one consider $H(x,y)=x/y$, then from Proposition \eqref{HVF} we obtain that $X_H = X$.

Note that $X$ is also a symplectic Hamiltonian vector field on the symplectic manifold $(U,\omega_{X_1,X_2})$ associated with the Hamiltonian $H(x,y)=x/y$, since $i_X \omega_{X_1,X_2}=dH$.
\end{example}

\subsection*{Acknowledgment}
This work was supported by a grant of the Romanian National Authority for Scientific Research, CNCS-UEFISCDI, project number PN-II-RU-TE-2011-3-0103.


\noindent {\sc R.M. Tudoran}\\
West University of Timi\c soara\\
Faculty of Mathematics and Computer Science\\
Department of Mathematics\\
Blvd. Vasile P\^arvan, No. 4\\
300223 - Timi\c soara, Rom\^ania.\\
E-mail: {\sf tudoran@math.uvt.ro}\\
\medskip

\end{document}